\documentclass{article}
\usepackage{doi}

\newcommand{\OpenDreamKit}{the
  \href{http://opendreamkit.org}{OpenDreamKit}
  \href{https://ec.europa.eu/programmes/horizon2020/}{Horizon 2020}
  \href{https://ec.europa.eu/programmes/horizon2020/en/h2020-section/european-research-infrastructures-including-e-infrastructures}{European Research Infrastructures} project (\#\href{http://cordis.europa.eu/project/rcn/198334_en.html}{676541})}

\newcommand{\GACI}{the French National Research Agency program
  (\href{http://www.agence-nationale-recherche.fr/ProjetIA-15-IDEX-0002}{ANR-15-IDEX-02})}

\usepackage{multirow,mathtools}
\usepackage{algorithm}
\usepackage{algpseudocode}
\usepackage{amsthm,amsfonts,amssymb}
\usepackage{arydshln}
\usepackage[capitalise,noabbrev,nameinlink]{cleveref}

\usepackage[T1]{fontenc}
\newtheorem{theorem}{Theorem}
\theoremstyle{remark}
\newtheorem{remark}{Remark}
\usepackage{xspace}

\newcommand{\nnz}[1]{\ensuremath{\##1}\xspace}
\newcommand{\F}{\mathsf{F}}
\newcommand{\bnd}[2]{\ensuremath{#1\mathopen{}\left(#2\right)\mathclose{}}}
\newcommand{\softoh}[1]{\bnd{\tilde{O}}{#1}}

\newcommand{\tmem}[1]{{\em #1\/}}

\renewcommand{\le}{\leqslant}
\renewcommand{\ge}{\geqslant}
\renewcommand{\leq}{\leqslant}
\renewcommand{\geq}{\geqslant}

\newcommand{\rest}{\clubsuit} 
\usepackage{xstring}
\newcommand*{\ExtractFirstChar}[1]{%
\StrRemoveBraces{#1}[\BracesRemoved]%
\StrChar{\BracesRemoved}{1}[\FirstChar]%
\FirstChar%
}
\newcommand{\Trsm}[3][]{\texttt{\ExtractFirstChar{#1}\ExtractFirstChar{#2}\ExtractFirstChar{#3}Trsm}}
\newcommand{\TrsmEC}[3][]{\Trsm[#1]{#2}{#3}\texttt{EC}}

\newenvironment{cmatrix}[1][]{\left[\begin{array}{#1}}{\end{array}\right]}

\usepackage{paralist}

\usepackage[colorinlistoftodos]{todonotes}

\title{LU factorization with errors\footnote{This work is partly
    funded by \OpenDreamKit~and~\GACI.}}

\author{Jean-Guillaume Dumas\footnote{
  {Universit\'e Grenoble Alpes}.
  {Laboratoire Jean Kuntzmann, CNRS, UMR 5224}.
  {700 avenue centrale, IMAG - CS 40700},
  {38058 Grenoble, cedex 9}
  {France}.
\href{mailto:Jean-Guillaume.Dumas@univ-grenoble-alpes.fr,Clement.Pernet@univ-grenoble-alpes.fr}{\{firstname.lastname\}@univ-grenoble-alpes.fr}}
\and Joris van der Hoeven\footnote{
  {CNRS}.
  {Laboratoire d'informatique de l'\'Ecole polytechnique}.
  {1, rue Honor\'e d'Estienne d'Orves, CS35003}.
  {LIX, UMR 7161 CNRS, 91120 Palaiseau, France}.
  \href{mailto:vdhoeven@lix.polytechnique.fr}{vdhoeven@lix.polytechnique.fr}}
\and Cl\'ement Pernet\footnotemark[2]
\and Daniel S. Roche\footnote{
  {United States Naval Academy}.
  {Annapolis, Maryland},
  {U.S.A.}
  \href{mailto:roche@usna.edu}{roche@usna.edu}}
}

\usepackage{color,svgcolor}
\usepackage{hyperref}
\makeatletter
\hypersetup{
  pdftitle={\@title},
  pdfauthor={Jean-Guillaume Dumas and Joris van der Hoeven and
    Cl\'ement Pernet and Daniel S. Roche},
  breaklinks=true,
  colorlinks=true,
  linkcolor=darkred,
  citecolor=blue,
  urlcolor=darkgreen,
}
\makeatother

\begin{document}
\maketitle

\begin{abstract}
We present new algorithms to detect and correct errors in the 
lower-upper factorization of a matrix, or the triangular linear system
solution, over an arbitrary field.
Our main algorithms do not require any additional information or
encoding other than the original inputs and the erroneous output.
Their running time is softly linear in the dimension times the number of
errors when there are few errors, smoothly growing to the cost of
fast matrix multiplication as the number of errors increases.
We also present applications to general linear system solving.
\end{abstract}

\section{Introduction}
The efficient detection and correction of computational errors is an
increasingly important issue in modern computing. Such errors can result from
hardware failures, communication noise, buggy software, or even
malicious servers or communication channels.

The first goal for fault-tolerant computing is
verification. Freivalds
presented a linear-time algorithm to verify the correctness of a single matrix
product~\cite{Freivalds:1979:certif}. Recently, efficient verification algorithms for a wide range of
computational linear algebra problems have been
developed~\cite{Kaltofen:2011:quadcert,jgd:2014:interactivecert,jgd:2016:gammadet,jgd:2017:rpmcert}.

Here we go further and try to correct any errors that arise.
This approach is motived by the following scenarios:

\paragraph{Large scale distributed computing}
In high-performance computing, the failure of
some computing nodes (fail stop) or
the corruption of some bits in main memory by cosmic radiation (soft
errors) become relevant.
The latter type can be handled by introducing redundancy and
applying classical error correction, either at the hardware
level (e.g., with ECC RAM), or integrated within the computation algorithm,
as in many instances of Algorithm Based Fault Tolerance
(ABFT)~\cite{HuAb84,Du:2011:softerrors,Davies:2013:CSEOLUF,Bouteiller:2015:ABFTLUQR}.  

\paragraph{Outsourcing computation}
When running some computation on one or several third-party
servers, incorrect results may originate from either failure or malicious
corruption of some or all of the third parties. The result obtained is then
considered as an approximation of the correct result.

\paragraph{Intermediate expression swell}
It often happens that symbolic computations suffer from requiring large
temporary values,
even when both the input and output of a problem are small or sparse.
It can then be more efficient to
use special methods that are able to determine or guess the sparsity pattern
and exploit this in the computation. Sparse polynomial and rational
function interpolation has
been developed \cite{Pro1795,BenOrTiwari1988,KaltofenYagati1988,Roche2018}
as such a technique. The connection with
error correction comes from the fact that we may regard a sparse
output as the perturbation of some trivial approximate solution, such as zero
or the identity matrix.

\paragraph{Fault-tolerant computer algebra}
Some recent progress on
fault tolerant algorithms has been made for Chinese
remaindering~\cite{Goldreich:2000:crterrors,Khonji:2010:ODRRS,Bohm:2015:badprimes},
system solving~\cite{Boyer:2014:NLSPEEC,Kaltofen:2017:ETPLSS}, matrix
multiplication and
inversion~\cite{Pagh:2013:compressed,Gasieniec:2017:freivalds,Roche:2018:ECFMMI},
and function recovery~\cite{Comer:2012:SPIBM,Kaltofen:2013:SMFR}.

\subsection{Our setting}

In this paper, we focus on LU-factorization and system solving. We will
assume the input (such as a matrix $A$ to be factored) is known, as well
as an \emph{approximate output} (such as a candidate LU-factorization)
which may contain errors. We seek efficient algorithms to
recover the correct outputs from the approximate ones.

Our work can
be seen as part of a wider effort to understand how techniques for fault tolerance
without extra redundancy extend beyond basic operations such as matrix multiplication. It
turns out that the complexities for other linear algebra problems are
extremely sensitive to the precise way they are stated. For instance, in the
case of system solving, the complexities for error-correction are quite
different if we assume the LU-factorization to be returned along with the
output, or not. The LU-factorization process itself is very sensitive to
pivoting errors; for this reason we will assume our input matrix to admit a
generic rank profile~(GRP).

From an information theoretic point of view,
it should be noted that all necessary input information
in order to compute the output is known to the client.
The approximate output is merely provided by the server in order to
make the problem computationally cheaper for the client.
In theory, if the client has sufficient resources, he could perform
the computation entirely by himself,
while ignoring the approximate output.
Contrary to what happens in the traditional theory of
error correcting codes, it is therefore not required to
transmit the output data in a redundant manner
(in fact, we might even not transmit anything at all).

In our setting, unlike classical coding theory, it is therefore more
appropriate to measure the correction capacity in terms of the amount
of computational work for the decoding algorithm rather than the
amount of redundant information.

We also put no {\em a priori} restriction on the number of errors:
ideally, the complexity should increase with the number of errors and
approach the cost of the computation without any approximate output when
the number of errors is maximal.

As a consequence, the error-correcting algorithms we envision
can even be used in an extreme case of the second scenario:
a malicious third party might introduce errors according to patterns
that are impossible to correct using traditional bounded distance
decoding approaches.  Concerning the third scenario,
the complexity of our error-correcting algorithms is typically sensitive
to the sparsity of a problem, both in the input {\em and\/} output.

Another general approach for error
correction is discussed in Section~\ref{sec:kns}: if we allow the server to
communicate some of the intermediate results, then many linear algebra
operations can be reduced in a simple way to matrix multiplication with error
correction.

\subsection{General outline and main results}

{\tmem{General notations.}} Throughout our paper, $\F$ stands for an effective
field with $\# \F \in \mathbb{N} \cup \{ \infty \}$ elements. We write $\F^{m
\times n}$ for the set of $m \times n$ matrices with entries in $\F$ and $\#M$
for the number of non-zero entries of a matrix $M \in \F^{m \times n}$. The
soft-Oh notation $\softoh{\ldots}$ is the same as the usual big-Oh notation
but ignoring sub-logarithmic factors: $f = \softoh{g}$ if and only if
$f{\in}O(g(\log{g})^{O(1)})$ for cost functions~$f$ and~$g$.

Let $\omega$ be a constant between $2$ and $3$ such that two $n \times n$
matrices can be multiplied using $O (n^{\omega})$.
In practice, we have $\omega = 3$ for small dimensions~$n$ and
$\omega \leqslant \log_2 7 \approx 2.81$ for larger dimensions,
using Strassen multiplication (the precise threshold depends on
the field $\F$; Strassen multiplication typically becomes interesting
for $n$ of the order of a few hundred). The best asymptotic
algorithm currently gives $\omega < 2.3728639$~\cite{LeGall:2014:fmm}.
From a practical point of view, the $\omega$ notation indicates whether
an algorithm is able to take advantage of fast low-level matrix multiplication
routines.\medskip

{\tmem{Generic solution.}} In \cref{sec:kns}, we first describe a generic strategy
for fault-tolerant computations in linear algebra. However, this strategy may
require the server to communicate certain results of intermediate
computations. In the remainder of the paper, we only consider a stronger form
of error correction, for which only the final results need to be
communicated.\medskip

{\tmem{LU-factorization.}} Let $A \in \F^{n \times n}$ be an invertible matrix
with generic rank profile (GRP). Let $\mathcal{L}, \mathcal{U} \in \F^{n \times n}$ be
(resp.) lower- and upper-triangular matrices that are ``close to'' the
LU-factorization of~$A$:
$A \approx \mathcal{L}\mathcal{U}$.
Specifically, suppose there exist sparse upper- and lower-triangular error
matrices $E, F \in \F^{n \times n}$ with
\[ A = (\mathcal{L} + E)  (\mathcal{U} + F), \qquad \#E +\#F \leqslant k. \]
 In the following we write the (possibly) faulty matrices in
 calligraphic fonts.
We specify that $\mathcal{L}$ has 1's on the diagonal, so that $E$ is strictly
lower-triangular whereas $F$ may have corrections on the diagonal or above it.

Given $A, \mathcal{L}, \mathcal{U}$, the goal is to determine the true
$L=\mathcal{L}+E$ and $U=\mathcal{U}+F$ as efficiently as
possible. Our main result is a probabilistic Monte Carlo algorithm for doing
this in time
\[ \softoh{t + \min\{kn, k^{\omega-2}n^{4-\omega}\}}, \]
for any constant probability of failure,
where $t$ is the number of nonzero terms in all input matrices,
$k$ is the number of errors in $\mathcal{L}$ or $\mathcal{U}$, and
$n$ is the matrix dimension (see \cref{thm:croutec}).
Here we measure the complexity in terms of
he number of operations in $\F$, while assuming that random elements in $\F$
can be generated with unit cost.
Note that because the number of errors $k$ cannot exceed the total dense
size $n^2$, this complexity is never larger than $\softoh{n^\omega}$,
which is the cost
of re-computing the LU-factorization with no approximate inputs, up to
logarithmic factors.
In \cref{sec:rect}, we also show how to generalize our
algorithm to rectangular, possibly rank deficient matrices (still under the
GRP assumption).\medskip

{\tmem{System solving.}} In \cref{sec:linsys}, we turn to the problem of solving a
linear system $XA = B$, where $A$ is as above and $B$ is also given.
Given only
possibly-erroneous solution $\mathcal{X}$ to such a system,
we do not know of any efficient method to correct the potential
errors. 
Still, we will show how errors can efficiently be corrected if we require
the server to return a few extra intermediate results.

More precisely, if $B$ has few rows, then we require an approximate
LU-factorization $A \approx \mathcal{L}\mathcal{U}$ and an approximate
solution to the triangular system $\mathcal{Y}\mathcal{U}= B$. If $B$ has many
rows (with respect to the number of columns), then we require an approximate
LU-factorization $A \approx \mathcal{L}\mathcal{U}$ and an approximate inverse
$\mathcal{R}$ of $\mathcal{U}$.
Given these data, we give algorithms to correct all errors in similar
running time as above, depending on the number of errors in all provided
inputs.

\section{Generic solution}\label{sec:kns}

In~\cite[\S~5]{Kaltofen:2011:quadcert} a generic strategy is described
for verifying the result of a linear algebra computation.
The complexity of the verification is the same as the complexity of
the actual computation under the assumption that matrix products
can be computed in time $\softoh{n^2}$.  In this section,
we show that a similar strategy can be used to correct errors.

Let $\mathcal{A}$ be any algorithm which uses matrix multiplication.
We assume that $\mathcal{A}$ is deterministic; any random bits needed in the computation should
be pre-generated by the client and included with the input.

The server
runs algorithm $\mathcal{A}$ and, each time it performs any matrix multiplication, it
adds that product to a list. This list of all intermediate products is sent back to the client.
We assume that there may be a small number of errors in any of the intermediate products,
but that these errors do not propagate; that is, the total number of erroneous entries in
any intermediate product is bounded by some $k$.

Next, to correct errors, the client also runs algorithm $\mathcal{A}$, except that every
time it needs to multiply matrices, it performs error correction instead, using the
intermediate result sent by the server. All other operations (additions, comparisons, etc.)
are performed directly by the client.

The total cost for the server is the same as the normal computation of $\mathcal{A}$ without
error correction. The cost for the client,
as well as the communication, is the cost that the algorithm \emph{would have}
if matrix multiplication could be performed in $O(n^2)$ time.

In particular, typical block matrix algorithms such as LU factorization
admit a worst case complexity of $\softoh{n^\omega}$ for the server and
only $\softoh{n^2}$ for the client (including communications).

The goal of this paper is to perform better, for instance when $k$ is
a bound on the number of errors only in $L$ and $U$ and not of all
intermediate computations.

\section{Block recursive algorithm}

We will now describe an error cleaning algorithm emulating the steps of an LU
decomposition algorithm, where each computation task is replaced by an error
cleaning task.

The cleaning algorithm to be used needs to satisfy the following properties:
\begin{compactenum}
\item \label{cond:block} it has to be a block algorithm, gathering most operations into
  matrix multiplications where error cleaning can be efficiently performed by
  means of sparse interpolation, as in~\cite{Pagh:2013:compressed,Roche:2018:ECFMMI};
\item \label{cond:recursive} it has to be recursive in order to make
  an efficient usage of fast matrix multiplication.
\item \label{cond:nointermediate} each block operation must be between
  operands that are submatrices of
  either the input matrix $A$ or the approximate $\mathcal{L}$ and $\mathcal{U}$
  factors. Indeed, the only data available for the error correction are these
  three matrices: the computation of any intermediate results that cannot directly
  be extracted from these matrices would be too expensive for achieving
  the intended complexity bounds.
\end{compactenum}

In the large variety of LU decomposition algorithms, iterative and block
iterative algorithms range in three main categories depending on the scheduling
of the update operation (see~\cite[\S~5.4]{DDSV98} and \cite{DGPS14} for a
description): right-looking, left-looking and Crout.

The right-looking variant updates the trailing matrix
immediately after a pivot or a block pivot is found, hence generating
blocks containing intermediate results, not
satisfying~(\ref{cond:nointermediate}). Differently, the left-looking and
the Crout variants proceed by computing directly the output coefficients in $L$
and $U$ following a column shape (left-looking) or arrow-head (Crout) frontier.
\begin{figure}[htb]
  \includegraphics[width=.25\columnwidth]{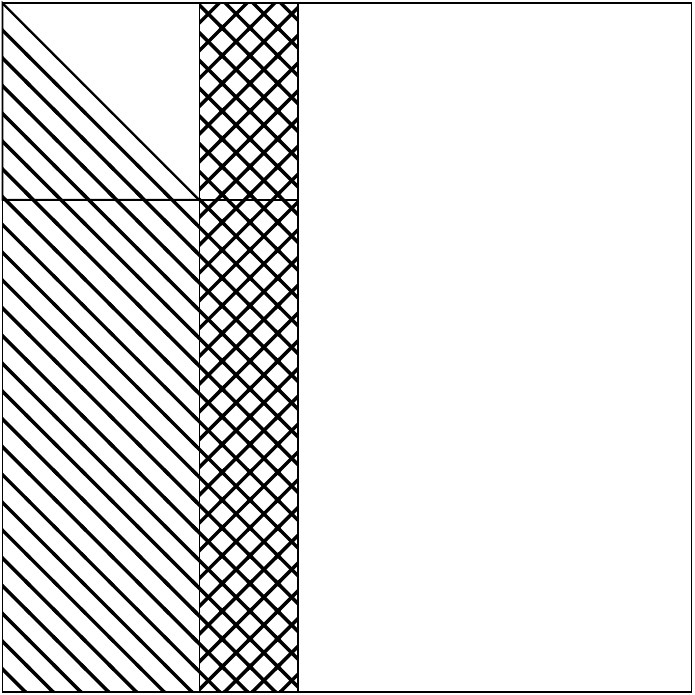}\hfill
  \includegraphics[width=.25\columnwidth]{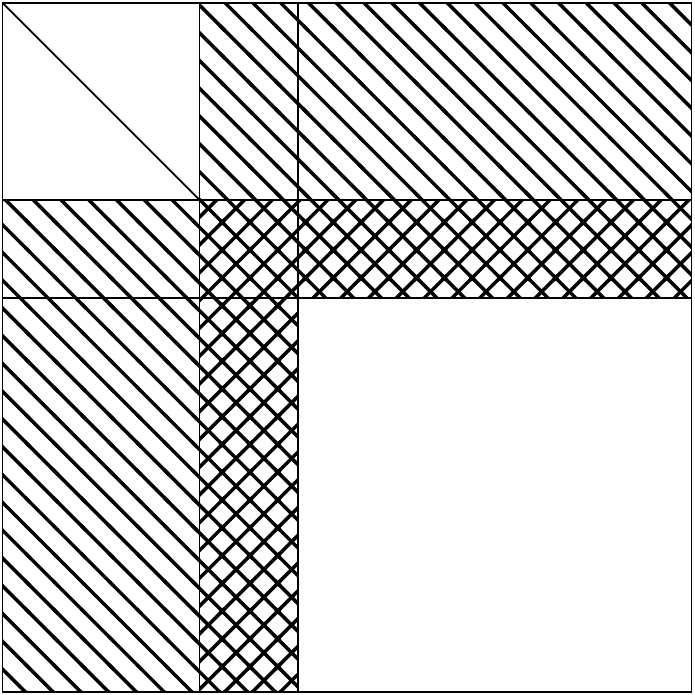}\hfill
  \includegraphics[width=.25\columnwidth]{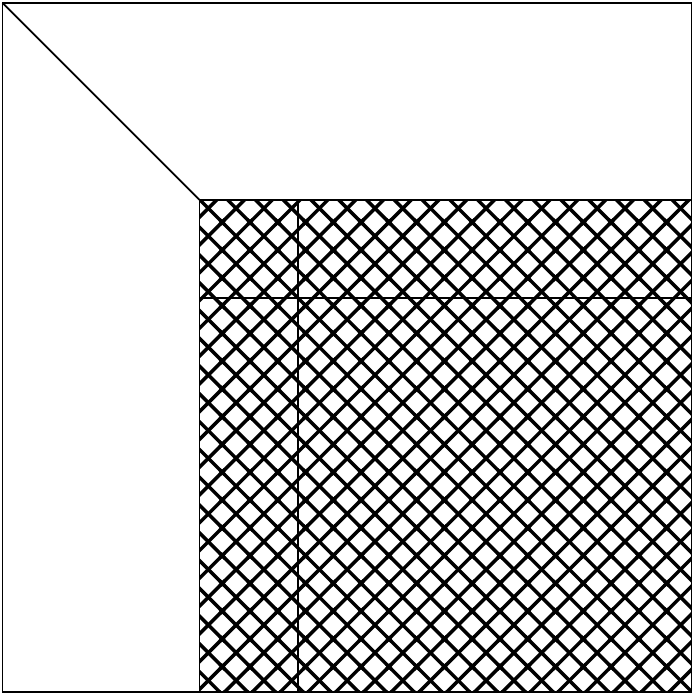}
  \caption{Access pattern of the left-looking (left), Crout (center) and
    right-looking variants of an LU factorization. Diagonal stripes represent
    read-only accesses, crossed stripes read-write accesses.}
  \label{fig:LLCroutRL}
  \end{figure}
Figure~\ref{fig:LLCroutRL} summarizes these 3 variants by exposing the memory
access patterns of one iteration.

The left-looking and the Crout schedules consist in delaying the computation of the Gauss updates
until the time where the elimination front deals with the location under
consideration. This is precisely satisfying
Condition~(\ref{cond:nointermediate}).
However these two schedules are usually described in an iterative or block
iterative setting. To the best of our knowledge, no recursive variant has been proposed so far.
We introduce in Section~\ref{sec:RecursiveCrout} a recursive Crout LU
decomposition algorithm on which we will base the error-correction algorithm of
section~\ref{sec:LUEC}. Interestingly, we could not succeed in writing a
recursive version of a left-looking algorithm preserving
Condition~(\ref{cond:nointermediate}).

\emph{We emphasize that, although our eventual error correction algorithm
will follow the recursive Crout variant, no assumption whatsoever
is made on the algorithm used (e.g., by a remote server)
to produce the approximate LU decomposition used as input.}

\subsection{Recursive Crout LU decomposition}
\label{sec:RecursiveCrout}

Algorithm~\ref{alg:crout} is a presentation of a Crout recursive variant of Gaussian elimination,
for a generic rank profile (GRP) matrix.
It incorporates the delayed update schedule of the classical block iterative
Crout algorithm~\cite{DDSV98,DGPS14} into a recursive algorithm dividing both
row and column dimensions.
Note that due to the delayed updates, the recursive algorithm need to be aware
of the coefficients in $L$ and $U$ previously computed, hence the entirety of
the working matrix has to be passed as input, contrarily to most recursive
algorithms~\cite{BuHo74, DPS:2013}.
In a normal context,
this algorithm could work in-place, overwriting the input matrix $A$
with both factors $L$ and $U$ as the algorithm proceeds.

In our error-correcting context, the input will eventually
consist not only of $A$ but also of the approximate $\mathcal{L}$ and
$\mathcal{U}$.  Therefore, in \cref{alg:crout} below,
we treat the matrix $A$ as an unmodified input and
fill in the resulting $L$ and $U$ into a separate matrix $M$.
In Algorithm~\ref{alg:croutec}, this matrix $M$ will initially contain the
approximate $\mathcal{L}$ and $\mathcal{U}$ which will be overwritten by the
correct $L$ and $U$.

The main work of the Crout decomposition consists of
dot products (in the base case), matrix multiplications and triangular solves.
We define TRSM as a triangular system solve with matrix
right-hand side, with some variants depending whether the triangular
matrix is lower ('L') or upper ('U') triangular, and
whether the triangular matrix
is on the left ('L', for $X\leftarrow T^{-1} A$)
or on the right ('R', for for $X\leftarrow A T^{-1}$).
These always work in-place, overwriting the right-hand side with the
solution.
For instance:
\begin{itemize}
\item $\Trsm{Upper}{Right}(A,U)$ is a right solve with upper-triangular
$U$ which transforms $A$ to $A'$ such that $A'U=A$.
\item $\Trsm{Lower}{Left}(L,B)$ is a left solve with lower-triangular
$L$ which transforms $B$ to $B'$ such that $LB'=B$.
\end{itemize}

In the algorithms, we use the subscript $\rest$ to denote ``indices 2
and 3'', so for example $n_\rest = n_2+n_3$ and
\[ A_{\rest\rest} =
  \begin{cmatrix}[c;{2pt/2pt}c]
    A_{22} & A_{23} \\\hdashline[2pt/2pt] A_{32} & A_{33}
  \end{cmatrix}.
\]

\begin{algorithm}
  \caption{\texttt{Crout}$(M,A_{\rest\rest},n_1,n_\rest)$}\label{alg:crout}
  \begin{algorithmic}[1]
    \Require{$M=\begin{cmatrix}[c|c] L_{11} \backslash U_{11} & U_{1\rest}
      \\   \hline L_{\rest 1} & \end{cmatrix}$ is
      $(n_1+n_\rest)\times(n_1+n_\rest)$\medskip}
    \Require{$L_{11}$ (resp. $U_{11}$)
    is unit lower (resp. upper) triangular}
    \Require{$A_{\rest\rest}$ is $n_\rest \times n_\rest$}
      \medskip
    \Require{
      $A = \begin{bmatrix} L_{11}U_{11} & L_{11}U_{1\rest}\\
        L_{\rest 1}U_{11} & A_{\rest\rest} \end{bmatrix}$ has GRP\medskip
    }
    \Ensure{$M= \begin{bmatrix} L \backslash U \end{bmatrix}$ such that
      $A=L{\cdot}U$\medskip
    }
    
    \If{$n_\rest=1$}
    \State\label{lin:croutdotprod}
      $U_{\rest\rest} \leftarrow A_{\rest\rest}-L_{\rest 1}{\cdot}U_{1\rest}$ \Comment{dot product}
      \newline
      \null\hfill
      \Comment{$L_{\rest\rest}$ is implicitly $[1]$}
    \Else
    \State Decompose $n_\rest = n_2 + n_3$
      with $n_2= \lceil{}n_\rest/2\rceil$, $n_3=\lfloor{}n_\rest/2\rfloor{}$
    \medskip

    \State Split
    $ M=\begin{cmatrix}[c|c;{2pt/2pt}c]
      L_{11}\backslash U_{11} & U_{12} & U_{13}\\
      \hline
      L_{21} &  & \\
      \hdashline[2pt/2pt]      L_{31} &  &
    \end{cmatrix}$
    \medskip

    \State Split
      \(A_{\rest\rest} =
        \begin{cmatrix}[c;{2pt/2pt}c]
          A_{22} & A_{23} \\\hdashline[2pt/2pt] A_{32} & A_{33}
        \end{cmatrix}
      \)
      \medskip

    \State\label{lin:croutrec1} $\text{\texttt{Crout}}\left(
      \begin{cmatrix}[c|c]L_{11}\backslash U_{11} & U_{12}\\\hline L_{21} & \end{cmatrix},
      A_{22}, n_1,n_2\right)$
    \medskip
      
    \State Here 
    $ M=\begin{cmatrix}[c|c;{2pt/2pt}c]
      L_{11}\backslash U_{11} & U_{12} & U_{13}\\
      \hline
      L_{21} & L_{22}\backslash U_{22}& \\
      \hdashline[2pt/2pt]      L_{31} & & 
    \end{cmatrix}$
    \medskip

    \State\label{lin:croutgemm1} $U_{23} \leftarrow A_{23}-L_{21}U_{13}$
    \State\label{lin:croutulltrsm}
    $\text{\Trsm{Lower}{Left}}(U_{23}, L_{22})$
      \Comment{$L_{22}U_{23} = A_{23}-L_{21}U_{13}$}
    \State\label{lin:croutgemm2} $L_{32} \leftarrow A_{32}-L_{31}U_{12}$
    \State\label{lin:crouturtrsm}
      $\text{\Trsm{Upper}{Right}}(L_{32},U_{22})$
    \Comment{$L_{32}U_{22} = A_{32}-L_{31}U_{12}$}
    \medskip

    \State Here $M=\begin{cmatrix}[cc|c]
        L_{11}\backslash U_{11} & U_{12} & U_{13}\\
        L_{21} & L_{22}\backslash U_{22} & U_{23}\\
        \hline
        L_{31} & L_{32} &
      \end{cmatrix}$
    \medskip
    
    \State $\text{\texttt{Crout}}
    \left(M, A_{33}, n_1+n_2,n_3\right) $
    \medskip

    \State Here $M=\begin{bmatrix}
      L_{11}\backslash U_{11} & U_{12} & U_{13}\\
      L_{21} & L_{22}\backslash U_{22} & U_{23}\\
      L_{31} & L_{32} & L_{33}\backslash U_{33} \end{bmatrix}$
    \EndIf
   \end{algorithmic}
\end{algorithm}

\begin{theorem}\label{thm:crout}
$Crout(M,A,0,n)$ overwrites $M$ with a complete LU factorization of
$A\in\F^{n{\times}n}$ in $O(n^\omega)$ operations.
\end{theorem}

\begin{proof}
Correctness is proven by induction on $n_\rest$.
For $n_\rest=1$ we have
\[
\begin{bmatrix} L_{11}&0\\ L_{2}&1\end{bmatrix}
\begin{bmatrix} U_{11}&U_{1\rest}\\0&A_{\rest\rest}-L_{\rest 1}U_{1\rest}\end{bmatrix}=
\begin{bmatrix} L_{11}U_{11}&L_{11}U_{1\rest}\\L_{\rest 1}U_{11}&A_{\rest\rest}\end{bmatrix}.
\]
For $n_\rest > 1$, we have
$$LU=
\begin{bmatrix}
L_{11}U_{11}&L_{11}U_{12}& L_{11}U_{13}\\
L_{21}U_{11}&L_{21}U_{12}+L_{22}U_{22}&L_{21}U_{13}+L_{22}U_{23}\\
L_{31}U_{11}&L_{31}U_{12}+L_{32}U_{22}&L_{31}U_{13}+L_{32}U_{23}+L_{33}U_{33}
\end{bmatrix}.
$$
After the first recursive call in step~\ref{lin:croutrec1},
we have $L_{21} U_{12} + L_{22} U_{22} = A_{22}$.
In step~\ref{lin:croutulltrsm}, we ensured that
$A_{23}=L_{22}U_{23}+L_{21}U_{13}$.
Similarly, $A_{32}=L_{32}U_{22}+L_{31}U_{12}$ from
step~\ref{lin:crouturtrsm}.
Finally, after the second recursive call,
we have $A_{33}=L_{31}U_{13}+L_{32}U_{23}+L_{33}U_{33}$.
This concludes the proof that $A=LU$ at the end of the algorithm.

The complexity bound stems from the fact that the dot products
in step~\ref{lin:croutdotprod} cost $O(n^2)$ overall,
and that the matrix multiplications and system
solves require $O(n^\omega)$~\cite{jgd:2008:toms}.
\end{proof}

\subsection{Error-correcting triangular solves}

The cost of \cref{alg:crout} is dominated by matrix-matrix products of
off-diagonal blocks of $L$ and $U$ and triangular solving in
steps~\ref{lin:croutgemm1}--\ref{lin:crouturtrsm}.
The first task in adapting this algorithm for error correction is to
perform error correction in this triangular solving step, treating the
right-hand side as an unevaluated \emph{black box matrix} in order to
avoid the matrix-matrix multiplication.

We explain the process to correct errors in $\mathcal{U}_{23}$
(an equivalent process works in the transpose for $\mathcal{L}_{32}$):
recall steps~\ref{lin:croutgemm1} and \ref{lin:croutulltrsm}
of \cref{alg:crout}:
 $U_{23} \leftarrow A_{23}-L_{21}U_{13}$
and
$\text{\Trsm{Lower}{Left}}(U_{23}, L_{22})$%
.
Mathematically, these steps perform the computation:
$U_{23} \gets L_{22}^{-1}\left(A_{23}-L_{21}U_{13}\right)$.
At this point in the error correction,
$A_{23}$ is part of the original input matrix while
$L_{21}$, $U_{13}$, and $L_{22}$ have already been corrected by the recursive
calls.

The idea is to be able to correct the next parts of an
approximate $\mathcal{U}$,
namely $\mathcal{U}_{23}$, without recomputing it.
\cref{alg:trsmec} below does this following the
approach of \cite{Roche:2018:ECFMMI}:
for $k\le mn$ total errors within $c\le n$ erroneous columns, less than
$c/2$ columns can have more than $s=\lfloor 2k/c \rfloor$ errors.
Therefore, we start with a candidate for $k$, then try to correct $s$
errors per erroneous column. If $k$ is correct, then they will be
corrected in fewer than $\log(c)\leq\log(n)$ iterations.
If fewer than $c/2$ columns are corrected on some step, this indicates
that the guess for $k$ was too low, so we double the guess for $k$ and
continue.  This is shown in
Algorithm~\ref{alg:trsmec}, where as previously mentioned, a normal
font is an already correct matrix block and a calligraphic font denotes
a matrix block to be corrected. We also recall that $\nnz{X}$
stands for the number of nonzero elements in $X$,
and define $\texttt{ColSupport}$ of a matrix $M$
to be the indices of columns of $M$ with any nonzero entries.

Due to the crucial use of sparse interpolation, we require a high-order
element $\theta$ in the underlying field $\F$. If no such $\theta$ exists,
we simply replace $\F$ by an extension field.  The cost of computing
in such an extension field will only induce a logarithmic overhead.

\begin{algorithm}
  \caption{\TrsmEC{Upper}{Right}$(\mathcal{R}, H, U, m, \ell, n,\varepsilon)$}\label{alg:trsmec}
  \begin{algorithmic}[1]
    \Require{$\mathcal{R}$ is $m\times n$}
    \Require{$H$ is $m\times n$ presented as an unevaluated blackbox $H=C-AB$
      where the inner dimension between $A$ and $B$ is $\ell$}
    \Require{$U$ is $n\times n$ invertible upper triangular}
    \Require{Failure bound $0<\varepsilon<1$.}
    \Ensure{$\mathcal{R}$ is updated in-place s.t. $\Pr[\mathcal{R}U=H] \ge 1-\varepsilon$\medskip}

    \State $k\leftarrow 1$ \Comment{how many errors exist}
    \State $k' \leftarrow 0$
    \Comment{how many errors have been corrected}
    \State \label{trsmec:algext}\textbf{if} $m \geq \#\F$ \textbf{then}
      $\F \leftarrow \mathbb{F}_{q^\nu}$ \, ($q=\#\F, \nu = \lceil \log_q (m+1) \rceil$) \textbf{end if}
    \State Pick $\theta$ of order $\ge m$ in $\F$, with precomputed $(\theta^{j})_{0\leq j < m}$
    \State $\lambda \leftarrow\left\lceil\log_{\#\F}\left({3n\log_2 n}/{\varepsilon}\right)\right\rceil$
    \State $c'\leftarrow 2n$
    \State $E \leftarrow 0^{m{\times}n}$
    \Repeat
    \State Pick $W\in \F^{\lambda \times m}$ uniformly at random.
    \State $X \gets WC - (WA)B$ \Comment{$X = WH$}
      \label{trsmec:colsup0}
    \State\label{trsmec:urtrsmone} \Trsm{Upper}{Right}$(X,U)$ \Comment{$X = WHU^{-1}$}
    \State $(j_1,\dots, j_{c})\leftarrow \texttt{ColSupport}(X - W\mathcal{R} - WE)$
      \newline \null\hfill \Comment{columns of $(\mathcal{R}+E)$ with errors}
      \label{trsmec:colsup}
    \State Clear any entries of $E$ from columns $j_1,\ldots,j_c$
    \State Update $\mathcal{R} \leftarrow \mathcal{R} + E$,\quad
      $k' \leftarrow k' + \nnz{E}$

    \State \textbf{if} $c > c'/2$ \textbf{then} $k \gets \max(2k,c)$
      \textbf{end if}\newline
      \null\hfill \Comment{too many errors; $k$ must be wrong}

    \State $c'\leftarrow c$
    \State $s\leftarrow \min\left(n, \left\lceil 2\frac{k-k'}{c}\right\rceil\right)$
    \State $V\leftarrow (\theta^{ij})_{0\leq i< 2s, 0\leq j < m}$
    \Comment{unevaluated}
    \State $P\leftarrow \begin{bmatrix} e_{j_1} & \dots & e_{j_c} \end{bmatrix}$
    \Comment{selects erroneous cols.\ of $\mathcal{R}$}
    \State $G \gets V(CP) -(VA)(BP) - (V\mathcal{R})(UP)$
    \State\label{trsmec:urtrsmtwo} \Trsm{Upper}{Right}$(G,P^\intercal U P)$
      \Comment{$GP^\intercal UP = V(H - \mathcal{R}U)P$}
    \State Find $S\in\F^{m{\times}c}$ s.t.\ $VS = G$ by sparse interpolation
    \State $E\leftarrow SP^{\intercal}$
    \Until{$c=0$}
  \end{algorithmic}
\end{algorithm}

\begin{theorem}\label{thm:trsmec}
  For a failure bound $0<\varepsilon<1$,
  $A\in\F^{m{\times}\ell}$, $B\in\F^{\ell{\times}n}$,
  $C\in\F^{m{\times}n}$, $U\in\F^{n{\times}n}$,
  $\mathcal{R}\in\F^{m{\times}n}$, 
  with total non-zero entries 
  $$t=\max\{\nnz{A},\nnz{B},\nnz{C},\nnz{U},\nnz{\mathcal{R}}\}\le(m+n)(\ell+n),$$
  and $k$ errors in $\mathcal{R}$,
  Algorithm~\ref{alg:trsmec} is correct and runs in time
  \[
  \softoh{(t+k+m)(1 + \log_{\#\F}\tfrac{1}{\varepsilon})+
    \max\{n,\ell\}\cdot \min\{k, k^{\omega-2} n^{3-\omega}\} }.
  \]
\end{theorem}

\begin{proof}
  Without loss of generality, we may assume that $m < \#\F$
  in step~\ref{trsmec:algext}. Indeed, in the contrary case,
  arithmetic in $\mathbb{F}_{q^\nu}$ is only $\softoh{\nu}=\softoh{\log m}$
  times more expensive than arithmetic in $\mathbb{F}_q$,
  which is absorbed by the soft-Oh of the claimed complexity bound.
  
  Define $R$ as the correct output matrix such that
  $RU=H$ and consider the beginning of some iteration through the loop.
  \Cref{trsmec:colsup} computes with high probability the column support
  of the remaining errors $R-(\mathcal{R}+E)$, viewing this as a
  blackbox $HU^{-1}-\mathcal{R}-E$.
  We project this blackbox on the left with a block of vectors $W$
  using a Freivalds check~\cite{Freivalds:1979:certif}.

  Hence $P$ is a $n\times{}c$ submatrix of a permutation matrix selecting the
  erroneous columns of $\mathcal{R}$. Selecting the same rows and
  columns in $U$ yields a $c\times c$ matrix $P^\intercal UP$ that
  is still triangular and invertible.
  With this, one can form a new blackbox
  $$S=(R-\mathcal{R})P=(H-\mathcal{R}U)\cdot P\cdot (P^{\intercal}UP)^{-1},$$
  using the fact that $(R-\mathcal{R})PP^\intercal=(R-\mathcal{R})$.
  The columns of this blackbox $S$ are viewed as
  $c$ sparse polynomials whose evaluation at powers of $\theta$ are
  used to recover them via sparse interpolation.

  Note that only the columns with at most $s$ nonzero entries are
  correctly recovered by the sparse interpolation, so some columns of
  $S$ may still be incorrect. However, any incorrect ones are discovered
  by the Freivalds check in the next round and never incorporated
  into~$\mathcal{R}$.
  From the definition of $s$, and the fact that sparse interpolation
  works correctly for all $s$-sparse columns of $R-\mathcal{R}$, we know
  that every iteration results in either $c$ reducing by half, or $k$
  doubling. Therefore the total number of iterations is at most
  $\log_2 c + \log_2 k \le (1+2)\log_2 n = 3 \log_2 n$.

  According to \cite[Lemma~4.1]{Roche:2018:ECFMMI}, the probability of
  failure in each Frievalds check in step~\ref{trsmec:colsup} is at most
  $\varepsilon/n$. By the union bound, the probability of
  failure at \emph{any} of the $\le 3 \log_2 n$ iterations
  is therefore at most $\varepsilon$, as required.

  Now for the complexity, the calls to compute the \texttt{ColSupport}
  on \crefrange{trsmec:colsup0}{trsmec:colsup} are all performed using
  sparse matrix-vector operations, taking $O({\lambda}(t+k))$.
  From~\cite[Lemma~6.1]{Roche:2018:ECFMMI}, the multiplication by the
  Vandermonde $2s\times n$ matrix $V$
  in $V(CP)$, $VA$, and $V\mathcal{R}$ all take
  $\softoh{t + k + m}$ operations.

  The recovery of $S$ by batched multi-sparse interpolation
  takes $\softoh{sc+m\log m}= \softoh{k+m}$
  operations~\cite[Theorem~5.2]{Roche:2018:ECFMMI}.

  What remains are the cost of computing the following:
  \begin{itemize}
    \item The product $(VA)(BP)$, which costs
      \(O(s\ell c / \min\{s,\ell,c\}^{3-\omega})\)
      using fast matrix multiplication.
    \item The product $(V\mathcal{R})(UP)$, which costs
      \(O(snc/\min\{s,n,c\}^{3-\omega}).\)
    \item The subroutine \Trsm{Upper}{Right}$(G,P^\intercal U P)$,
      which costs
      the same as it would be to multiply $G$ times $P^\intercal UP$,
      \(O(sc^2/\min\{s,c\}^{3-\omega})\).
  \end{itemize}
  From the definition of $s$ we have $sc\in O(k)$. 
  Let $N=\max\{n,\ell\}$. Since $s,c\le n$ and $n,\ell\le N$,
  all three costs are 
  \begin{equation}\label{eqn:ectrsm}
    O(Nk/\min\{s,c\}^{3-\omega}).
  \end{equation}
  Until the algorithm terminates, we always have $s,c\ge 1$, so
  \eqref{eqn:ectrsm} is at most $O(Nk)$, proving the first part
  of the $\min$ in the complexity.

  For the second part, observe that the number of erroneous columns $c$ must satisfy
  $k/n \le c \le n$, which means that
  $\min\{s,c\}\ge k/n$ by the definition of $s$.
  Then the cost in \eqref{eqn:ectrsm} is bounded above by
  $O\left(\frac{Nk}{\left(k/n\right)^{3-\omega}}\right)
    \le O\left(N \cdot k^{\omega-2} \cdot n^{3-\omega}\right)$.
\end{proof}

\begin{remark}
  Since the input matrix $U$ is not modified by the algorithm,
  some entries of $U$ may be defined implicitly --- in particular, if
  the matrix is unit diagonal and the 1's are not explicitly stored.
\end{remark}

\begin{remark}\label{rq:LLeqUR}
Transposing the algorithm, we may also correct $L\mathcal{R}=G$,
where $L$ is lower triangular:
\[
\TrsmEC{Lower}{Left}(\mathcal{R}, G, L, n, \ell, m, \varepsilon)
  =\left(\TrsmEC{Upper}{Right}(\mathcal{R}^{\intercal},G^{\intercal},L^{\intercal}, m, \ell, n, \varepsilon)\right)^{\intercal}.
\]
\end{remark}

\begin{remark}\label{rq:UpperLeft}
There are two more triangular variants to consider.
Computing \TrsmEC{Lower}{Right} (or, with Remark~\ref{rq:LLeqUR},
\TrsmEC{Upper}{Left}) could be done exactly the same as in
\cref{alg:trsmec}, except that the two subroutine calls,
lines~\ref{trsmec:urtrsmone} and~\ref{trsmec:urtrsmtwo}, would be
to \Trsm{Lower}{Right}.
\end{remark}

\subsection{Correcting an invertible LU decomposition}
\label{sec:LUEC}
We now have all the tools to correct an LU decomposition.
We suppose that the matrix $A$ is non-singular and has generic rank
profile (thus there exist unique $L$ and $U$ such that $A=LU$).
We are given possibly faulty candidate matrices $\mathcal{L}$,
$\mathcal{U}$ and want to correct them: for this we run
Algorithm~\ref{alg:crout}, but replace
lines~\ref{lin:croutgemm1}-\ref{lin:crouturtrsm} by two calls to
Algorithm~\ref{alg:trsmec}. 
The point is to be able to have explicit submatrices of $A$, $L$, $U$,
$\mathcal{L}$ and $\mathcal{U}$ for the two recursive
calls (no blackbox, nothing unevaluated there), and that all the base cases
represent only a negligible part of the overall computations (the
base case is the part of the algorithm that is recomputed explicitly
when correcting).
This is presented in Algorithm~\ref{alg:croutec}.

\begin{algorithm}
  \caption{\texttt{CroutEC}$(M,A_{\rest\rest},n_1,n_\rest,\varepsilon)$}\label{alg:croutec}
  \begin{algorithmic}[1]
    \Require{$M=\begin{cmatrix}[c|c] L_{11} \backslash U_{11} & U_{1\rest}
      \\\hline  L_{\rest 1} &
      \mathcal{L}_{\rest\rest}\backslash \mathcal{U}_{\rest\rest} \end{cmatrix}$
    is $(n_1+n_\rest)\times(n_1+n_\rest)$\medskip
    }
    \Require{$L_{11}$ (resp. $U_{11}$)
    is unit lower (resp. upper) triangular}
    \Require{$\mathcal{L}_{\rest\rest},\mathcal{U}_{\rest\rest}$
      are $n_\rest\times n_\rest$ unit lower/upper triangular\medskip}
    \Require{$A_{\rest\rest}$ is $n_\rest \times n_\rest$}
    \Require{
      $A = \begin{bmatrix} L_{11}U_{11} & L_{11}U_{1\rest}\\ L_{\rest 1}U_{11} & A_{\rest\rest} \end{bmatrix}$ has GRP\medskip
    }
    \Require{Failure bound $0<\varepsilon<1$\medskip}
    \Ensure{$M= \begin{bmatrix} L \backslash U \end{bmatrix}$ such that
      $\Pr[A=L{\cdot}U] \ge 1-\varepsilon$\medskip
    }
    
    \If{$n_\rest=1$}
    \State\label{line:pivot} $U_{\rest\rest} \leftarrow A_{\rest\rest}-L_{\rest 1}{\cdot}U_{1\rest}$
    \Comment{dot product}\newline
    \hspace*{\fill}\Comment{$L_{\rest\rest}$ is implicitly $[1]$, must be correct}
    \Else
    \State Decompose $n_\rest = n_2 + n_3$
      with $n_2= \lceil{}n_\rest/2\rceil$, $n_3=\lfloor{}n_\rest/2\rfloor{}$
    \medskip
    \State Split
    $M=\begin{cmatrix}[c|c;{2pt/2pt}c]
      L_{11}\backslash U_{11} & U_{12} & U_{13}\\ \hline
      L_{21} & \mathcal{L}_{22}\backslash\mathcal{U}_{22} & \mathcal{U}_{23}
        \\ \hdashline[2pt/2pt]
      L_{31} & \mathcal{L}_{32} & \mathcal{L}_{33}\backslash\mathcal{U}_{33}
    \end{cmatrix}$\medskip
    \medskip

    \State Split
      \(A_{\rest\rest} =
        \begin{cmatrix}[c;{2pt/2pt}c]
          A_{22} & A_{23} \\\hdashline[2pt/2pt] A_{32} & A_{33}
        \end{cmatrix}
      \)
      \medskip

 \State $\text{\texttt{CroutEC}} \left(\begin{cmatrix}[c|c]
    L_{11} \backslash U_{11} & U_{12} \\\hline
    L_{21} & \mathcal{L}_{22}\backslash\mathcal{U}_{22}\end{cmatrix},
    A_{22},n_1,n_2,\varepsilon/4\right)$
    \label{croutec:rec1}
    \medskip

    \State Here
    $M=\begin{cmatrix}[c|c;{2pt/2pt}c]
      L_{11}\backslash U_{11} & U_{12} & U_{13}\\ \hline
      L_{21} & L_{22}\backslash U_{22} & \mathcal{U}_{23}
        \\ \hdashline[2pt/2pt]
      L_{31} & \mathcal{L}_{32} & \mathcal{L}_{33}\backslash\mathcal{U}_{33}
    \end{cmatrix}$\medskip
    \medskip

   \State $\text{\TrsmEC{Lower}{Left}}(\mathcal{U}_{23}, A_{23}-L_{21}U_{13}, 
      L_{22}, n_2, n_1, n_3,\varepsilon/4)$
      \label{croutec:trsm1}
      \newline
    \hspace*{\fill}\Comment{$A_{23}-L_{21}U_{13}$ is left unevaluated}
    \medskip

    \State $\text{\TrsmEC{Upper}{Right}}(\mathcal{L}_{32},
      A_{32}-L_{31}U_{12},U_{22},
      n_3, n_1, n_2,\varepsilon/4)$
      \label{croutec:trsm2}
      \newline
    \hspace*{\fill}\Comment{$A_{32}-L_{31}U_{12}$ is left unevaluated}
    \medskip

    \State Here $M=\begin{cmatrix}[cc|c]
        L_{11}\backslash U_{11} & U_{12} & U_{13}\\
        L_{21} & L_{22}\backslash U_{22} & U_{23}\\
        \hline
        L_{31} & L_{32} & \mathcal{L}_{33}\backslash\mathcal{U}_{33}
      \end{cmatrix}$
    \medskip

    \State $\text{\texttt{CroutEC}} \left( M, A_{33}, n_1+n_2, n_3,
      \varepsilon/4\right) $
      \label{croutec:rec2}
    \medskip

    \State Now $M=\begin{bmatrix}
      L_{11}\backslash U_{11} & U_{12} & U_{13}\\
      L_{21} & L_{22}\backslash U_{22} & U_{23}\\
      L_{31} & L_{32} & L_{33}\backslash U_{33}\end{bmatrix}$.
\EndIf
   \end{algorithmic}
\end{algorithm}

\begin{theorem}\label{thm:croutec}
  For $A\in\F^{n{\times}n}$ which has GRP,
  $\mathcal{L}\in\F^{n{\times}n}$ unit lower triangular
  and
  $\mathcal{U}\in\F^{n{\times}n}$ upper triangular,
  with total non-zero entries
  $t=\nnz{A}+\nnz{\mathcal{L}}+\nnz{\mathcal{U}}$,
  failure bound $0<\varepsilon<1$,
  and $k$ errors in $\mathcal{L}$ and $\mathcal{U}$,
  Algorithm \texttt{CroutEC}$(\begin{bmatrix}\mathcal{L}\backslash \mathcal{U}\end{bmatrix},A,0,n,\varepsilon)$
  is correct and runs in time
  \[
  \softoh{(t+k)(1+\log_{\#\F}\tfrac{1}{\varepsilon})+ \min\{kn,k^{\omega-2}n^{4-\omega}\}}.
  \]
\end{theorem}

\begin{proof}
  Correctness follows from Theorems~\ref{thm:crout}
  and~\ref{thm:trsmec}: we rewrite Algorithm~\ref{alg:crout}, but
  replace the intermediate two matrix multiplications and two \Trsm{}{}s
  by two calls, with blackboxes, to Algorithm~\ref{alg:trsmec}.
  Passing $\varepsilon/4$ to all subroutines ensures that the total
  probability of failure in any Frievalds check in any \TrsmEC{}{}
  is at most $\varepsilon$. Note that the shrinking $\varepsilon$ does
  not affect the soft-oh complexity, since at the bottom level we will
  have $\varepsilon'=\varepsilon/(4^{\log_2(n)})=\varepsilon/n^2$, and
  $\log\tfrac{1}{\varepsilon'}$ is
  $O(\log\tfrac{1}{\varepsilon}+\log n)$.

  For the complexity, note that since $A$ has GRP, $\nnz{A}\ge n$.
  The stated cost bound depends crucially on the following fact:
  in each level of recursive calls to \texttt{CroutEC},
  each call is correcting a single diagonal block
  $\mathcal{L}_{\rest\rest}\backslash\mathcal{U}_{\rest\rest}$
  and uses only the parts of $M$ and $A$ above and left of that block.

  This claim is true by inspection of the algorithm: the blocks
  $\mathcal{L}_{22}\backslash\mathcal{U}_{22}$ and
  $\mathcal{L}_{33}\backslash\mathcal{U}_{33}$ being corrected
  in the two recursive calls to \texttt{CroutEC} on
  \cref{croutec:rec1,croutec:rec2} are clearly disjoint diagonal blocks.
  And we see also that the algorithm never uses the top-left part of $M$,
  namely $L_{11}{\backslash}U_{11}$; all arguments to the calls to \TrsmEC{}{} on
  \cref{croutec:trsm1,croutec:trsm2} are (disjoint) submatrices above and left of
  $\mathcal{L}_{\rest\rest}\backslash\mathcal{U}_{\rest\rest}$, as
  shown, e.g., in Figure~\ref{fig:proofcrout}.

\begin{figure}[ht]\hfill
  \begin{minipage}[c]{0.45\columnwidth}
    \includegraphics[width=\columnwidth]{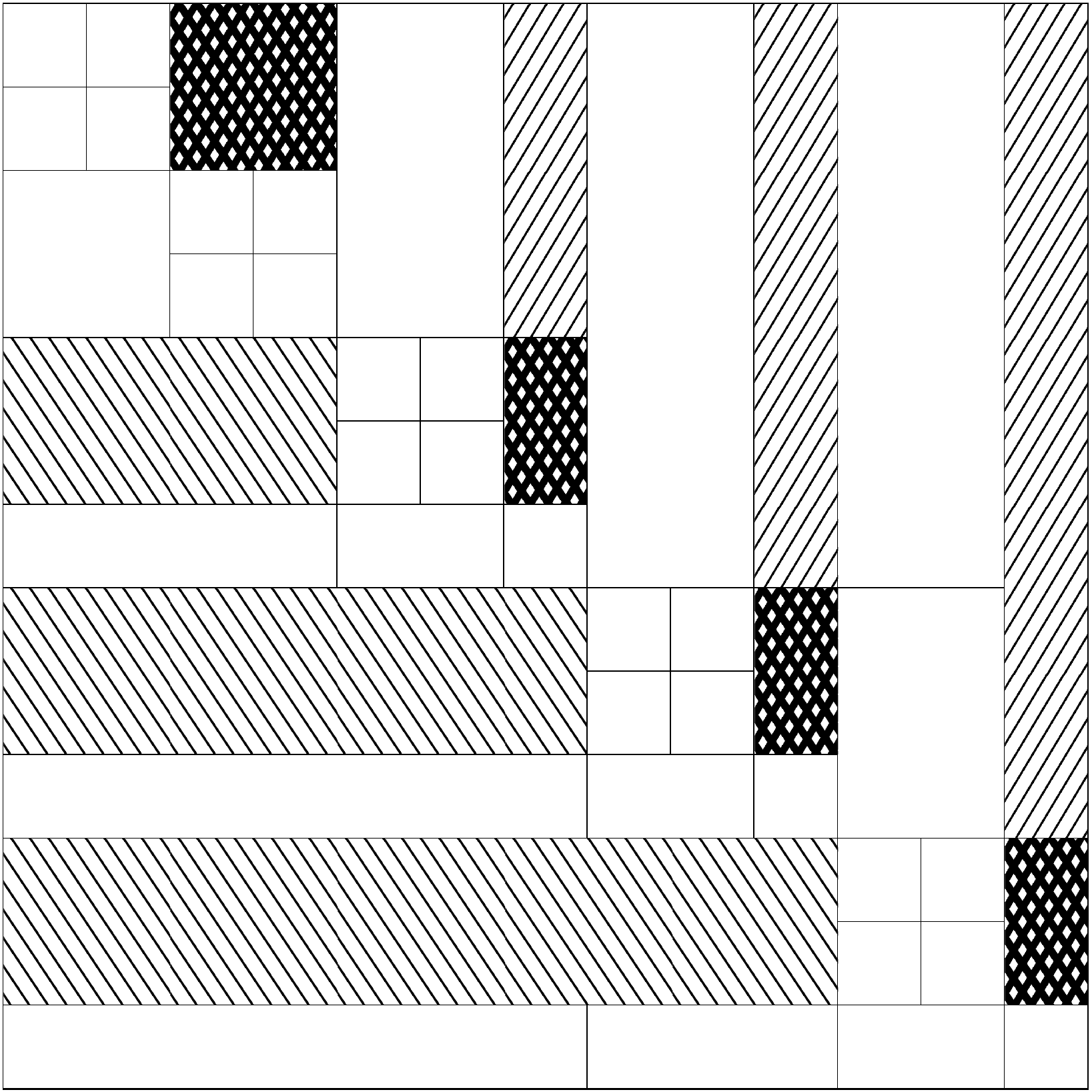}
  \end{minipage}\hfill
  \begin{minipage}[c]{0.45\columnwidth}
    \caption{All updates at the third recursion level at
      step~\ref{croutec:trsm1}. The figure indicates the
      locations of $\mathcal{U}_{23}$, $L_{21}$, and $U_{13}$.}
    \label{fig:proofcrout}
  \end{minipage}\hfill
\end{figure}

  With this understanding, we can perform the analysis.
  The work of the algorithm is entirely in the dot products in the base case,
  and the calls to \TrsmEC{}{} in the recursive case.

  For the base case dot products, these are to correct the diagonal entries of
  $\mathcal{U}$, using the diagonal of $A$ and disjoint, already-corrected rows of
  $L$ and columns of $U$. Therefore the total cost of the dot products is
  $O(t+k)$.

  For the rest, consider the $i$th recursive level of calls to \texttt{CroutEC},
  where $0\le i < \log_2 n$.
  There will be exactly $2^{i+1}$ calls to \TrsmEC{}{} on this level, whose inputs
  are all disjoint, from the claim above. Write
  $m_{ij}, \ell_{ij},$ etc.\ for the parameters
  in the $j$th call to \TrsmEC{}{} on level $i$, for $1\le j\le 2^{i+1}$.

  Every call to \TrsmEC{}{} on the $i$th level satisfies:
  \begin{itemize}
    \item $m_{ij},n_{ij} \in O(n/2^i)$;
    \item $\ell_{ij},N_{ij}=\max\{n_{ij};\ell_{ij}\} \le n$;
    \item $\varepsilon_{ij} = \varepsilon/4^i$, so that
      $\lceil\log_{\#\F}(1/\varepsilon_{ij})\rceil$ is
      $O((1 + \log_{\#\F}\tfrac{1}{\varepsilon})+\log n)$;
    \item $\sum_{j=1}^{2^{i+1}} t_{ij} \le t$
      because the submatrices are disjoint at the same recursive level; and
    \item $\sum_{i=0}^{\log_2 n}\sum_{j=1}^{2^{i+1}} k_{ij} \le k$
      because each error is only corrected once.
  \end{itemize}
  Now we wish to compute the total cost of all calls to \TrsmEC{}{}, which by
  \cref{thm:trsmec} and the notation just introduced, is soft-oh of
  \[
    \sum_{i=0}^{\log_2 n}
    \sum_{j=1}^{2^{i+1}}
      \left[
      (t_{ij}+k_{ij}+m_{ij})(1 + \log_{\#\F}\tfrac{1}{\varepsilon_{ij}})
      + N_{ij} \cdot \min\{k_{ij}, k_{ij}^{\omega-2} n_{ij}^{3-\omega}\}
      \right].
  \]

  Taking the first term of the sum, this is
  \begin{align*}
    &\sum_{i=0}^{\log_2 n}
    \sum_{j=1}^{2^{i+1}}
      (t_{ij} + k_{ij} + \tfrac{n}{2^i})(1 + \log_{\#\F}\tfrac{1}{\varepsilon}+\log n)
    \\ &\le
    O\left((t\log n + k + 2n\log n)(1 + \log_{\#\F}\tfrac{1}{\varepsilon}+\log n)\right)
    \\ &=
    \softoh{(t + k + n)(1 + \log_{\#\F}\tfrac{1}{\varepsilon})},
  \end{align*}
  which gives the first term in our stated complexity.

  The second term of the sum simplifies to
  \begin{equation}\label{eqn:croutec2}
    n\cdot
    \sum_{i=0}^{\log_2 n}
    \sum_{j=1}^{2^{i+1}}
      \min\{k_{ij}, k_{ij}^{\omega-2} (n/2^i)^{3-\omega}\}.
  \end{equation}

  Now observe that, by definition, each individual summand is less
  than or equal to both parts of the $\min$ expression. Therefore we bound the sum of
  minima by the minimum of sums; that is, the previous summation is at most
  \[
    n\cdot\min\left\{
      {\textstyle \sum_i \sum_j} k_{ij},\quad
      {\textstyle \sum_i \sum_j} k_{ij}^{\omega-2} (n/2^i)^{3-\omega}
    \right\}.
  \]

  Because each error is only corrected once, $\sum_i \sum_j k_{ij} \le k$.
  Using H\"older's inequality and $0 \le \omega-2 < 1$, this yields
  \[
    \sum_{j=1}^{r} k_{ij}^{\omega-2} \le r \left(\frac{\sum_{j=1}^r
    k_{ij}}{r}\right)^{\omega-2} \le r^{3-\omega} k^{\omega-2}, \]
  for all $r$.  Applying this to the second summation above gives
  \begin{align*}
    \sum_{i=0}^{\log_2 n} \sum_{j=1}^{2^{i+1}}
      k_{ij}^{\omega-2} (n/2^i)^{3-\omega}
    & =
      n^{3-\omega}
      \sum_{i=0}^{\log_2 n}
      2^{(\omega-3)i}
      \sum_{j=1}^{2^{i+1}}
      k_{ij}^{\omega-2}
    \\ & \le
      n^{3-\omega}
      \sum_{i=0}^{\log_2 n}
      2^{(\omega-3)i}
      \cdot
      2^{(i+1)(3-\omega)}
      \cdot
      k^{\omega-2}
    \\ & \le
      O(k^{\omega-2} n^{3-\omega}\log n).
  \end{align*}

  Then the entirety of \eqref{eqn:croutec2} becomes just
  \(\softoh{\min\{kn, k^{\omega-2} n^{4-\omega}\}},\)
  which gives the second term in the stated complexity.
\qedhere\end{proof}

\subsection{Correcting a rectangular, rank-deficient LU}
\label{sec:rect}

For $m \leq n$, assume first that \(A=[A_1\quad A_2]\)
is a rectangular $m{\times}n$ matrix such that
$A_1$ is square $m{\times}m$ with GRP.
Assume also that
\([\mathcal{L} \backslash \mathcal{U}_1\quad  \mathcal{U}_2]\)
is an approximate LU decomposition.  Then we may correct potential errors as follows:
\begin{compactenum}
\item
  $\text{\texttt{CroutEC}}([\mathcal{L}_1\backslash \mathcal{U}_1],A_1,0,n,\varepsilon/2)$;
  \hspace*{\fill} \Comment{corrects $L_1 \backslash U_1$}
\item
  $\text{\TrsmEC{Lower}{Left}}(\mathcal{U}_2,A_2,L_1,m,m,n-m,\varepsilon/2)$.
  \hspace*{\fill} \Comment{corrects $U_2$}
\end{compactenum}

Second, if $A_1$ is rank deficient, but still GRP, then the first occurrence of
a zero pivot $U_{\rest\rest}$, line~\ref{line:pivot} in Algorithm~\ref{alg:croutec},
reveals the correct rank: at this point $L_{11}$, $U_{11}$, and the
upper (resp. left) part of $L_{\rest 1}$ (resp. $U_{1\rest}$) are correct. 
It is thus sufficient to stop the elimination there, and recover the
remaining parts of
\[ L=\begin{bmatrix} L_{11}\\L_{\rest 1}\end{bmatrix}\in\F^{m{\times}r}, \quad
   U=\begin{bmatrix} U_{11}&U_{1\rest}\end{bmatrix}\in\F^{r{\times}n}, \]
using (corrected) triangular system solves, as follows:
\begin{compactenum}
\item Let $r=\textrm{rank}\; A$ and
  $A=\begin{bmatrix} A_{11}&A_{1\rest}\\ A_{\rest 1} & A_{\rest\rest}\end{bmatrix}$,
  where $A_{11}$ is $r{\times}r$.
\item \texttt{CroutEC}$(\begin{bmatrix}\mathcal{L}_{11}\backslash\mathcal{U}_{11}\end{bmatrix},
  A_{11},0,r,\varepsilon/3)$
\item
$\text{\TrsmEC{Lower}{Left}}(\mathcal{U}_{1\rest},A_{1\rest},L_{11},r,r,n-r,\varepsilon/3)$.
  \hspace*{\fill} \Comment{corrects $U_{1\rest}$}
\item $\text{\TrsmEC{Upper}{Right}}(\mathcal{L}_{\rest 1},A_{\rest
  1},U_{11},m-r,r,r,\varepsilon/3)$.
  \hspace*{\fill} \Comment{corrects $L_{\rest 1}$}
\end{compactenum}

\section{System solving}\label{sec:linsys}
One application of LU factorization is linear system solving.
Say $A$ is $n{\times}n$ invertible, and $B$ is $m{\times}n$.
Ideally, correcting a
solution $\mathcal{X}$ to the linear system $XA=B$, should depend only
on the errors in~$\mathcal{X}$.
Indeed, our algorithm \TrsmEC{Upper}{Right} does exactly this in the special
case that $A$ is upper-triangular.

Unfortunately, we do not know how to do this for a general non-singular
$A$ using the previous
techniques, as we for instance do not know of an efficient way to
compute the nonzero columns of the erroneous entries in
$(X-\mathcal{X}) = BA^{-1}-\mathcal{X}$.

By adding some data to the solution $\mathcal{X}$ (and therefore,
unfortunately, potentially more errors), there are several
possibilities:

\begin{itemize}
\item Generically, a first solution is to proceed as in
  Section~\ref{sec:kns}. One can use \cite{Kaltofen:2011:quadcert}, but
  then the complexity depends not only on errors in $\mathcal{X}$, but
  on errors in all the intermediate matrix products.
\item A second solution is to invert the matrix $A$ using
  \cite[Algorithm~6]{Roche:2018:ECFMMI} (\texttt{InverseEC}) and then
  multiply the right-hand side using
  \cite[Algorithm~5]{Roche:2018:ECFMMI} (\texttt{MultiplyEC}).
  This requires some extra data to correct, namely a candidate inverse
  matrix, and the complexity depends on errors appearing now both in the
  system solution and in that inverse candidate matrix as follows:
\begin{enumerate}
\item $Z\leftarrow{}\texttt{InverseEC}(A,\mathcal{Z})$; \hfill\(\triangleright\){$Z=A^{-1}$}
\item $X\leftarrow{}\texttt{MultiplyEC}(B,Z,\mathcal{X})$. \hfill\(\triangleright\){$X=BA^{-1}$}
\end{enumerate}
\end{itemize}

Now, computing an inverse, as well as correcting it, is more expensive
than using an LU factorization: for the computation itself, the inverse
is more expensive by a constant factor of $3$ (assuming classic matrix arithmetic), and
for \texttt{InverseEC} the complexity
of~\cite[Theorem~8.3]{Roche:2018:ECFMMI} requires the fast selection
of linearly independent rows using~\cite{Cheung:2013:rank}, which might be
prohibitive in practice.

For these reasons, we prefer to solve systems using an LU factorization.
The goal of the remainder of this section is to do this
with a similar complexity as for \cref{alg:croutec} and
while avoiding to rely on the sophisticated algorithm
from~\cite{Cheung:2013:rank}.

\subsection{Small right-hand side}\label{ssec:srhs}
An intermediate solution, requiring the same amount of
extra data as the version with the inverse matrix, but using only fast
routines, can be as follows. Use as extra data a candidate
factorization $\mathcal{LU}$ and a candidate intermediate right-hand
side $\mathcal{Y}$ of dimension $m{\times}n$, such that
$\mathcal{Y},\mathcal{U}$ are approximations to
the true $Y,U$ with $YU=B$. We simultaneously correct
errors in $\mathcal{X}$, $\mathcal{L}$, $\mathcal{U}$,
and $\mathcal{Y}$ as follows:
\begin{enumerate}
\item \texttt{CroutEC}$(\mathcal{L}\backslash \mathcal{U}, A,
  0, n, \varepsilon/3)$; \hfill\(\triangleright\){\; $L$ and $U$ with $A=LU$}
\item
$\TrsmEC{Upper}{Right}(\mathcal{Y},B,U,m,0,n,\varepsilon/3)$; \hfill\(\triangleright\){\; $Y$ with $YU=B$}
\item
$\TrsmEC{Lower}{Right}(\mathcal{X},Y,L,m,0,n,\varepsilon/3)$. \hfill\(\triangleright\){\; $X$ with $XL=Y$}
\end{enumerate}
Note, of course, that if the number of rows $m$ in $B$ is very small,
say only $m\le n^{o(1)}$,
then it is faster to recover $\mathcal{L}$ and $\mathcal{U}$ only, by
running \texttt{CroutEC}, and then compute $Y$ and $X$ directly from the
corrected $L$ and $U$ with classical TRSMs.

\subsection{Large right-hand side}
If the row dimension $m$ of $B$ is large with respect to the column
dimension $n$,
then the matrix $\mathcal{Y}$ from above will be larger than
$\mathcal{U}$. The client
can instead ask the server to provide as extra data $\mathcal{R}$ as a candidate for
$U^{-1}$, to correct it with $U$, and then to
use $L$ and $U^{-1}$ to correct $\mathcal{X}$ directly:
\begin{enumerate}
\item $\texttt{CroutEC}(\mathcal{L}\backslash
\mathcal{U},A,0,n,\varepsilon/3)$; \hfill\(\triangleright\){\; $L$ and
  $U$ with $A=LU$}
\item $\texttt{TrInvEC}(\mathcal{R},U,n,\varepsilon/3)$;
  \hfill\(\triangleright\){$R=U^{-1}$}
\item \TrsmEC{Lower}{Right}$(\mathcal{X}, BR, L, m, n, n,
\varepsilon/3)$,
  \hfill\(\triangleright\){$XL=BR=BU^{-1}$}
\end{enumerate}
with \texttt{TrInvEC} a variant of \texttt{InverseEC} sketched below as
\cref{alg:trinvec}
(where TRSV is a triangular system solve with a
column vector and TRMV is a matrix-vector multiplication).
This does not require the expensive algorithm
of~\cite{Cheung:2013:rank} to select independent columns, as the
matrix is triangular.

Note that, in the call to \TrsmEC{Lower}{Right} in
the last step, the right-hand side $BR$ is left unevaluated, just as in
the calls to \TrsmEC{}{} from \cref{alg:croutec}.

\begin{algorithm}
  \caption{\texttt{TrInvEC}$(U,\mathcal{R},n,\varepsilon)$ corrects $\mathcal{R}+E=U^{-1}$}\label{alg:trinvec}
  \begin{algorithmic}[1]
    \State
    $P_J=\texttt{ColSupport}(U^{-1}v-\mathcal{R}v)$;\Comment{$U^{-1}v$ via
        TRSV with $U$, $\mathcal{R}v$ via TRMV with candidate}
    \State $T=( P^T_J U P_J)^{-1}$;\Comment{$O(r^\omega)$}
    \State $E'=(EP_J)=(I-\mathcal{R}U)P_J T^{-1}$, so it can
    be recovered via multi-sparse interpolation as 
    $V E' = \left(VP_J-(V\mathcal{R})(UP_J)\right)T^{-1}$.
  \end{algorithmic}
\end{algorithm}

\section{Conclusion}
We have shown how to efficiently correct errors in an LU factorization,
and how to apply this error correction to system solving.

A few remaining challenges are to:
\begin{itemize}
\item Generalize our error-correcting algorithms to matrices $A$ which do not
  have GRP, correcting more general factorizations such as
  $A=PLUQ$ where $P,Q$ are permutation matrices. Our approach works
  directly if the permutations $P$ and $Q$ are known to be error-free,
  but correcting erroneous permutations $\mathcal{P}$ and $\mathcal{Q}$
  is more difficult.
\item Directly correct errors in $\mathcal{L}$ and $\mathcal{R}$
  such that $AR=L$, i.e., $R=U^{-1}$. This would be useful for system
  solving, as we have seen above.
\item More generally, correcting errors only in the solution $X$ of a
  linear system, without any extra information from the server, would be
  an even more ambitious goal.
\end{itemize}

\bibliographystyle{plainurl}
\bibliography{lubib}

\end{document}